\documentclass[pra,reprint,twocolumn,notitlepage]{revtex4-1}
\usepackage{amsmath}
\usepackage[latin1]{inputenc}
\usepackage{placeins} % FloatBarrier

\usepackage{graphicx}
\usepackage{subfigure}

\usepackage{amsthm}
\newtheorem{theorem}{Theorem}
\newtheorem{definition}{Definition}
\usepackage{dsfont}
\newcommand{\1}{\mathds{1}}

\DeclareMathOperator{\In}{In}
\DeclareMathOperator{\Out}{Out}
\newcommand{\ramuno}{\mathrm{i}}
\newcommand{\eulere}{\mathrm{e}}

\newcommand{\ket}[1]{| #1 \rangle}
\newcommand{\bra}[1]{\langle #1 |}
\newcommand{\braket}[2]{\langle #1 | #2 \rangle}

\usepackage{hyperref}
\hypersetup{
	colorlinks,
	citecolor=black,
	filecolor=black,
	linkcolor=black,
	urlcolor=black
}

\begin{document}
\title{Quantum Router with Network Coding}
\author{Michael \surname{Epping}}
\email{epping@hhu.de}
\author{Hermann \surname{Kampermann}}
\author{Dagmar \surname{Bru\ss{}}}
\affiliation{Institut f\"ur Theoretische Physik, Heinrich-Heine-Universit\"at D\"usseldorf, Germany}
\begin{abstract}
Many protocols of quantum information processing, like quantum key distribution or measurement-based quantum computation, "consume" entangled quantum states during their execution. When participants are located at distant sites, these resource states need to be distributed. Due to transmission losses quantum repeater become necessary for large distances (e.g. $\gtrsim 300\text{ km}$).\\
Here we generalize the concept of the graph state repeater to $D$-dimensional graph states and to repeaters that can perform basic measurement-based quantum computations, which we call quantum routers. This processing of data at intermediate network nodes is called quantum network coding. We describe how a scheme to distribute general two-colorable graph states via quantum routers with network coding can be constructed from classical linear network codes. The robustness of the distribution of graph states against outages of network nodes is analysed by establishing a link to stabilizer error correction codes. Furthermore we show, that for any stabilizer error correction code there exists a corresponding quantum network code with similar error correcting capabilities.

\end{abstract}
\maketitle
\section{Introduction}
Long distance quantum communication suffers from transmission losses. Hence intermediate devices that recover the original signal, so-called quantum repeaters, are necessary~\cite{Briegel98,Duer99}. Many proposals for them have been made, including approaches based on repeat-until-success strategies using two-way communication~\cite{Duan01,vanLoock06,Zwerger12} and forward-error correction based protocols which do not require this acknowledgement of successful transmission~\cite{Knill96,Jiang09,Fowler10,Muralidharan14}. The requirements of quantum repeaters regarding the precision of operations are very challenging, but experiments have shown significant progress \cite{Cory98,Yuan08,Kimble08,Sangouard11,Schindler11,Clausen12,Hensen15}.\\
In \cite{Epping16} we described a quantum repeater scheme based on quantum graph states~\cite{Schlingemann01,Hein04}, which naturally generalizes to networks of such devices. Such a network consists of several parties connected by repeater lines according to a mathematical graph. The aim is to provide the parties with a resource of entangled quantum states, which can be used in several protocols of quantum information theory, e.g. key distribution~\cite{BB84,Chen04,Duer05}, Bell experiments~\cite{Mermin1990}, teleportation~\cite{BBC93}, distributed quantum computing~\cite{Buhrman03}, and secret sharing~\cite{Markham08,Keet10}.\\
A recently introduced research field studies network codes, i.e. networks in which the intermediate nodes are able to process the data packets instead of merely routing them~\cite{Ahlswede00,Yeung08,Ho08}. Network coding can improve the throughput as well as increase the robustness of the communication scheme against outages of nodes. Analogous studies have been performed in the quantum case, i.e. regarding quantum network codes (QNC) \cite{Leung06,Hayashi07,Kobayashi09,Kobayashi10,Beaudrap14}. Again one can distinguish them by the way classical communication is treated.\\
Here we consider network coding in the scenario of free classical communication, which is reasonable given that it is usually easier to implement and well-developed in practice. While we described the error correction on the physical layer (i.e. single qubits) in \cite{Epping15,Epping16}, we now focus on analysing quantum error correction on the higher level of the network, where errors correspond to failures of whole nodes. As the error model we use the depolarizing channel, i.e. lost qudits are replaced by the completely mixed state, which corresponds to some probability for discrete $X$ and $Z$ errors on the qudit of the lost node~\cite{nielsen00}.\\
In this paper we describe how network coding, implemented here by simple measurement-based quantum computation~\cite{Raussendorf03}, helps to distribute multipartite entangled qudit graph states amongst distant nodes of a network. This task is a general prerequisite for many quantum information protocols, see above.\\
We start with a short introduction to qudit graph states in Section~\ref{sec:graphstates}. In the same section we summarize the main idea of the graph state repeater~\cite{Epping16} and generalize the concept from qubits to qudits (i.e. $D$-dimensional quantum systems with $D\geq 2$). In Section~\ref{sec:networkcoding} we apply network coding to the repeater network starting from the examples of the butterfly network and linear network codes and describe different codes that lead to the same distributed graph state. The main result of this paper, the analysis of the robustness of the quantum network coding scheme in terms of stabilizer error correction codes~\cite{Gottesman97,LidarBrunQEC13}, is developed in Section~\ref{sec:robustness}. We end with an outlook and conclusions in Section~\ref{sec:conclusion}.

\section{Graph states and quantum networks}\label{sec:graphstates}
In this work we make extensive use of the stabilizer formalism~\cite{Gottesman97,nielsen00}. Instead of directly writing down the quantum state $\ket{\psi}$ of a system it can be more convenient to keep track of a set $S$ of commuting operators, which uniquely define $\ket{\psi}$ via the eigenequations
\begin{equation}
s\ket{\psi}=\ket{\psi}\hspace{1cm}\forall s\in S. \label{eq:eigenequation}
\end{equation}
We say $s$ stabilizes $\ket{\psi}$. Please note that products of stabilizer operators fulfil Eq.~(\ref{eq:eigenequation}), too.\\
Consider a measurement of the observable $\mathcal{A}$ with outcome $\alpha$ and let $P_{\alpha}$ denote the projector onto the eigenspace of $\mathcal{A}$ associated with the eigenvalue $\alpha$ by $P_{\alpha}$. One may update $S$ corresponding to the new knowledge by replacing all $s\in S$ that commute with $P_{\alpha}$ by $P_{\alpha} s$. All stabilizer operators that do not commute with $P_{\alpha}$ are removed from $S$ after the measurement. Let the post-measurement state $\ket{\psi'}$ be defined by a minimal set of operators that commute with $P_{\alpha}$ (for all $\alpha$). We call these operators the main stabilizers. Identifying them suffices to understand the measurement dynamics.\\
We will be mostly interested in the non-Fd part $B$ of a system composed of $A$ and $B$ after measuring the observable $\mathcal{A}$ on $A$ with outcome $\alpha$. We therefore trace out system $A$ after the measurement. If a stabilizer operator $s$ has the form $\mathcal{A}^x\otimes \mathcal{B}$, with integer $x$, before the measurement of $\mathcal{A}$, then we simply replace it by $\alpha^x\mathcal{B}$ (acting on $B$ only) after the measurement.\\
The stabilizer formalism is particularly handy in the context of graph states~\cite{Schlingemann01,Hein04,Hostens04}. A qudit graph state $\ket{G}$ associated with a mathematical graph $G$ is a particular quantum state of the composite system of several $D$-dimensional quantum systems (the vertices), see e.g.~\cite{Hostens04}. We denote the corresponding graph by $G=(V,E)$, where $V$ is the set of vertices and $E$ is the set of edges. Each edge $(a,b)\in E$ has an integer-valued weight, which can be found at the position ($a$, $b$) in the adjacency matrix $\Gamma$, i.e. the weight is $\Gamma_{ab}$. The graph state $\ket{G}$ is defined to be the unique quantum state stabilized by the so-called stabilizer generators
\begin{align}
g_v=& X_v \prod_{w\in V} Z_w^{\Gamma_{vw}}\hspace{1cm} \forall v\in V,\label{eq:graphgenerators}\\
 \text{where }X=&\sum_{j=0}^{D-1} \ket{j+1\bmod D}\bra{j}\\
\text{and }
 Z=&\sum_{j=0}^{D-1} \eulere^{j \frac{2 \pi \ramuno}{D}} \ket{j}\bra{j}.
\end{align}
Because $X^D=Z^D=\1$ many calculations in the remainder will be done modulus $D$, i.e. in the finite field $\mathds{F}_D$. For example $Z^j=Z^{j\bmod D}$ and $X^{-1}=X^{D-1}$.
And more explicitly, 
\begin{equation}
\ket{G}=\prod_{\substack{a,b\in V\\a>b}} (C^{(a,b)}_Z)^{\Gamma_{ab}} \ket{+}^{\otimes |V|},
\end{equation}
where
\begin{equation}
\ket{+}=\frac{1}{\sqrt{D}} \sum_{j=0}^{D-1} \ket{j}
\end{equation}
and the controlled-Phase gate acting on qudit $a$ and $b$ is
\begin{equation}
C_Z^{(a,b)} = \sum_{j=0}^{D-1} \ket{j}\bra{j}_a \otimes Z_b^j \otimes \1_{V\backslash\{a,b\}}.
\end{equation}
$X$ and $Z$ are not Hermitian (except for $D=2$) but normal and we will use them as observables, i.e. we measure in the eigenbasis of the operator, projecting onto the eigenspace associated with an eigenvalue. We assign this eigenvalue to the measurement outcome. The discrete Fourier transform matrix
\begin{equation}
 H=\frac{1}{\sqrt{D}} \sum_{xy} \eulere^{xy \frac{2\pi \ramuno}{D}} \ket{x}\bra{y}
\end{equation}
performs the basis change $H X H^\dagger=Z$ and $H Z H^\dagger=X^{-1}$.

The above considerations regarding measurements in the stabilizer formalism help to understand the idea of the graph state repeater of \cite{Epping16}. 

In the present context a quantum network consists of several (many) separated sites, which are connected by quantum channels. We keep the physical level with the error correction code against transmission losses hidden from our abstraction, because it will not be important in the present paper. See \cite{Epping16,Epping15} for the omitted details. We may therefore assume that the quantum channels are perfect.\\

The task is to produce a graph state shared by a subset of these sites, which we call ``parties''. If the other sites have vertex degree two, then they are called quantum repeaters, else they are called quantum routers. The latter term is used to distinguish them from the standard concept of quantum repeaters in the literature. Repeaters and routers will be depicted by squares, while parties are shown as disks. The quantum circuits of a quantum repeater and a quantum router are shown in FIG.~\ref{fig:circuits}. \begin{figure}[htbp] %
\subfigure[A quantum repeater.]{\includegraphics[scale=0.4]{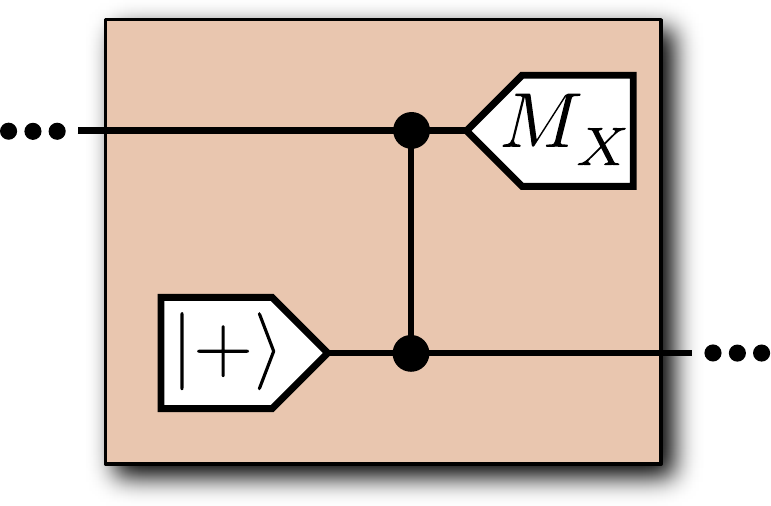}}\hspace{0.5cm} %
\subfigure[A quantum router with vertex degree five.]{\includegraphics[scale=0.4]{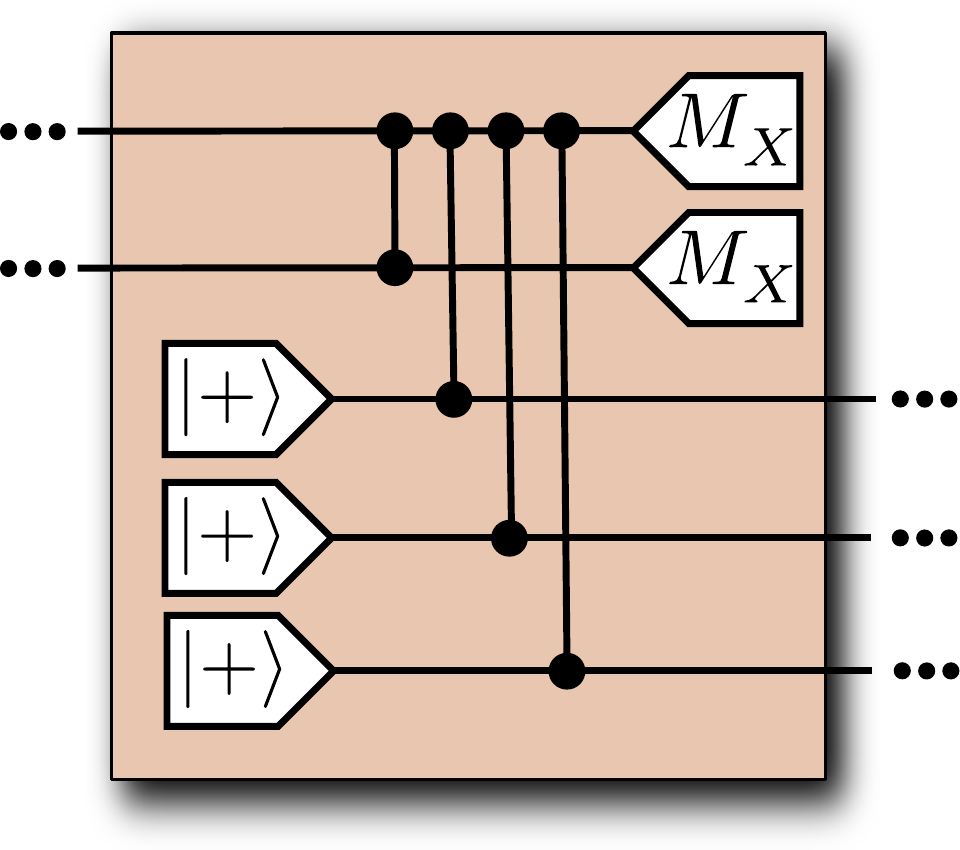}} %
\caption{The circuits of a repeater (vertex degree two) and a router (vertex degree greater than two), see also~\cite{Epping16,Epping15}. The shown state preparation, controlled-Phase gate and measurement are meant to be logical, i.e. the physical level is hidden from abstraction.}\label{fig:circuits} %
\end{figure} %
A quantum network allows to generate a graph state like the one shown in FIG.~\ref{fig:repeaternetwork}. \begin{figure}[tbp] %
\subfigure[The pre-measurement graph state.]{\includegraphics[scale=0.9]{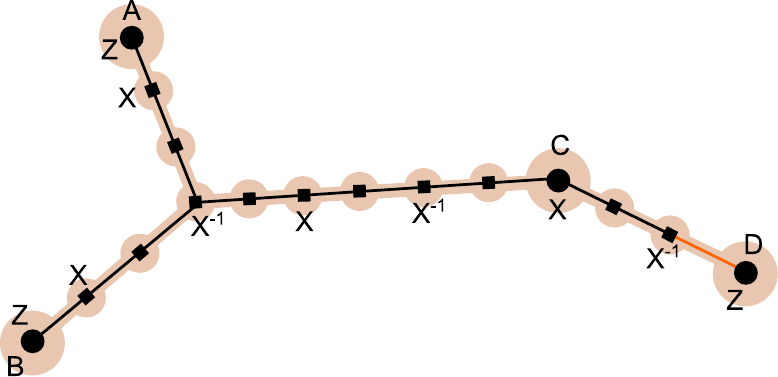}\label{fig:repeaternetworka}}\\
\subfigure[The post-measurement graph state.]{\includegraphics[scale=0.9]{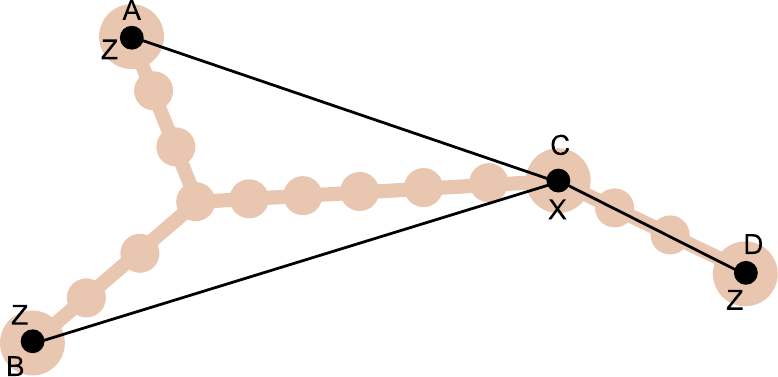}\label{fig:repeaternetworkb}}
\caption{A network of repeater stations (squares) connects the parties $V=\{A,B,C,D\}$ (disks). The orange (light grey) edge has weight $-1$ (s.t. no basis change is necessary for the final state), black edges have weight one unless explicitly labelled otherwise. All repeaters and routers (squares) measure their qudit in $X$-basis. The main stabilizer centred on $C$ shown in (a) defines the neighbourhood of $C$ in the post-measurement state (b) shared by the parties A, B, C, and D. The brown background indicates the network infrastructure, i.e. the sites and fibres.}\label{fig:repeaternetwork}%
\end{figure} %
The transmission directions on the channels to create the edges can always be assigned, s.t. the network is acyclic and the nodes can be brought into a time-order, where each site receives all qudits before it sends the outgoing qudits. An edge of weight $\gamma\in\mathds{Z}$ between sites $A$ and $B$ is created as follows. Let $a$ be the qudit of site $A$, which might have been received from a previous site. A qudit $b$ is prepared at $A$ in the $\ket{+}$ state. Then $(C_Z^{(a,b)})^{\gamma}$ is applied to $a$ and $b$. Finally the qudit $b$ is sent to the site $B$.

Let a network corresponding to the graph $G=(V,E)$ be given, i.e. parties $V$ connected by repeater lines $E$. Suppose the graph state $\ket{G}$ is to be distributed. This can be achieved by producing the graph state with a qudit at each repeater station as shown in FIG.~\ref{fig:repeaternetwork}.

In our protocol all repeaters and routers measure their qudits in the $X$-basis. The post-measurement state can be found by looking at the main stabilizer operators, which are constructed from chains of $X$-operators obtained by multiplying powers of the generators of Eq.~(\ref{eq:graphgenerators}), where the exponents are chosen such that the $Z$-operators cancel out each other~\cite{Wu15}. Note that the weight of the edge enters as a multiplicative factor in the exponent of the $Z$-operators, see Eq.~(\ref{eq:graphgenerators}). As described above one can measure the qudits of the repeater stations in $X$-basis and replace the operators in the main stabilizer by the corresponding measurement outcomes. One obtains, up to by-product operators, the stabilizer generators of the desired graph state. This way one easily sees which graph state is produced.\\

\section{Network coding}\label{sec:networkcoding}
In \cite{Epping16} all repeater stations had vertex degree two. Now we consider repeaters (``routers'') with higher vertex degree. This leads to quantum network coding (QNC), i.e. the intermediate sites now perform a special type of (measurement-based) quantum computation instead of merely passing on the data. This can be advantageous in terms of the possible throughput (distributed states per network use) and the robustness against failures of nodes. A simple and famous example for the first advantage is the butterfly network~\cite{Ahlswede00,Leung06}, see FIG.~\ref{fig:qbutterfly}. \begin{figure}[htbp]%
\subfigure[Pre-measurement state.]{\hspace{0.5cm}\includegraphics{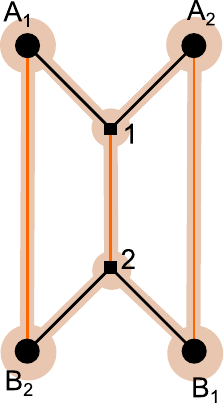}\hspace{0.5cm}} %
\subfigure[Post-measurement state.]{\hspace{0.5cm}\includegraphics{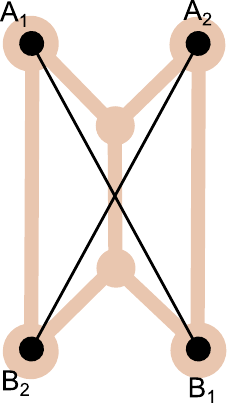}\hspace{0.5cm}} %
 \caption{The butterfly network. $A_1$ and $A_2$ want to share a Bell pair with $B_1$ and $B_2$, respectively. Orange (light grey) edges have weight $-1$ and black edges have weight one unless explicitly labelled otherwise.}\label{fig:qbutterfly}%
\end{figure}%
The desired target state has two edges: one from $A_1$ to $B_1$ and one from $A_2$ to $B_2$, each with edge weight $1$, i.e. $A_i$ and $B_i$ want to share a Bell pair. The communication on each network link is restricted to a single qudit. This task is impossible with a routing strategy~\cite{Satoh16}. From the main stabilizer operators $X_{A_1} X_2 Z_{B_1}$, $X_{A_2} X_2 Z_{B_2}$, $X_{B_1} X_1 Z_{A_1}$ and $X_{B_2} X_1 Z_{A_2}$ it is clear, that the desired state is produced up to by-product operators (see Appendix~\ref{app:byproduct}) by $X$ measurements on qudits $1$ and $2$.\\
The example of the butterfly network is actually translated from classical network coding theory (see below). We point out that a quantum network code like the one discussed previously can be found for any classical linear network code, see \cite{deBeaudrap14} and Theorem~\ref{thm:codeconstruction}.

\subsection{Classical network coding}
We focus on wired networks, which consist of nodes connected by communication lines (see FIG.~\ref{fig:butterfly}). The nodes can be distinguished into three classes: sources ($A$) send information to sinks ($B$), possibly via intermediate sites ($C$). Sources and sinks will be denoted parties. A classical network can be represented mathematically by a graph $G_{cl.}=(V_{cl.},E_{cl.})$, where $V_{cl.}$ is the set of nodes or vertices and $E_{cl.}$ is the set of links or directed edges. It will be convenient to associate imaginary edges with the sources (sinks) that do not have a starting point (endpoint) and transmit the message created (decoded) at the source (sink). Each edge $e\in E_{cl.}$ has the form $(a,b)$ and implies that a single symbol can be transmitted from $a$ to $b$ in a single network use. The underlying alphabet contains $D$ different symbols and will be denoted as the base field $\mathds{F}_D$. $G_{cl.}$ is required to be acyclic, which implies that nodes and the transmissions on the edges can be brought into a time order, such that each node received the incoming packets before it sends the outgoing ones. With this order of the nodes and edges at hand, we define for each node $v$ the vectors $\In(v)$ and $\Out(v)$ to contain the parents and children of $v$, respectively. In a slight abuse of the notation we denote the in- and out-degree of $v$ by $|\In(v)|$ and $|\Out(v)|$, respectively.\\
The instructions for each node how to process the packets form the network code. We focus on linear network codes meaning that the outgoing packets are linear combinations of the incoming packets. Even though situations in which non-linear codes outperform the linear ones are known~\cite{Dougherty05}, the latter nevertheless form a very important class of network codes. They are proven to be optimal in networks with a single source. Furthermore the linear nature allows for a simple decoding.\\
Following Raymond W. Yeung~\cite{Yeung08} we describe a linear network code over base field $\mathds{F}_D$ by its local and global encoding kernels. The global encoding kernel $\vec{f}_e$ of an edge $e$ gives the coefficients of the packet transmitted across $e$ in the basis of the source packets. For $|A|$ symbols created by sources the network code is called $|A|$-dimensional and $\vec{f}_e\in \mathds{F}_D^{|A|}$. The local encoding kernels are $|\Out(v)| \times |\In(v)|$-matrices $K_v$ for each node, that define the linear combination of the outgoing packets. The local encoding kernel $K_v$ at vertex $v$ relates the global encoding kernels of the incoming and outgoing edge via
\begin{equation}
 \vec{f}_{(v,\Out(v)_i)}=\sum_j {(K_v)}_{ij} \vec{f}_{(\In(v)_j,v)}.
\end{equation}
The notation is exemplified with the (classical) butterfly network in FIG.~\ref{fig:butterfly}. A linear network code on the Butterfly can outperform the best routing strategy.
\begin{figure}[tp]%
\subfigure[The global encoding kernels.]{\includegraphics[scale=0.9]{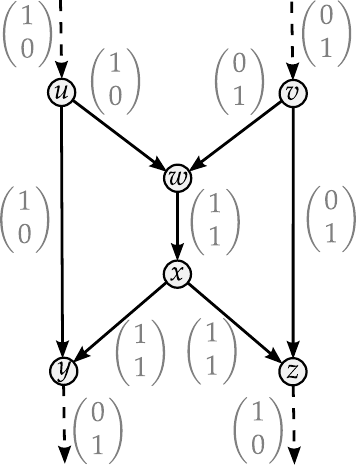}}\hspace{1cm} %
\subfigure[The local encoding kernels.]{\includegraphics[scale=0.9]{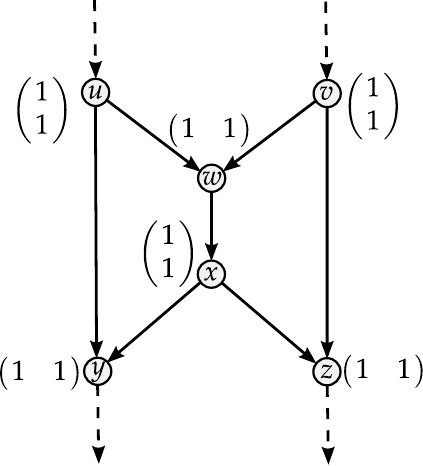}} %
\caption{The Butterfly network. It consists of two sources $A=\{u,v\}$, two sinks $B=\{y,z\}$ and two intermediate nodes $C=\{w,x\}$. The task is for $u$ to transmit a bit to $z$, while simultaneously $v$ casts a bit to $y$. Under the restriction that each link has capacity one, this task is not achievable with a routing strategy. But the depicted linear network code achieves it. Dashed edges are imaginary.}\label{fig:butterfly} %
\end{figure} % 
The overall effect of any linear network code with sources $A$ and sinks $B$ can be described by the $|B|\times |A|$-matrix $F$ that maps the vector of symbols sent by the sources to the vector of symbols received at the sinks. The network coding decomposes this matrix into a product of block matrices, the local encoding kernels, such that the rate constraints of the edges are fulfilled.\\

\subsection{Distribution of quantum two-colorable graph states}
In this section we describe quantum network codes that distribute graph states associated with two-colorable (bipartite) graphs. A two-colorable graph $G=(V,E)$ allows to split $V$ into two sets $V_1$ and $V_2$, such there are no edges between vertices of the same set. The construction is the same as the one of~\cite{Beaudrap14}, but we use the network for the distribution of graph states instead of teleportation only. The following graph state will be a useful building block.
\begin{definition}[Graph state for a linear map]\label{def:GK}
Let an arbitrary linear function $K:\mathds{F}_D^n\rightarrow \mathds{F}_D^m$ be given via its $m\times n$-matrix representation $K$. We define the graph state $\ket{K}$ via the two-colorable graph $G_K$ with partitions $A=\{A_1,A_2,...,A_n\}$ and $B=\{B_1,B_2,...,B_m\}$, i.e. vertices $V=\{A,B\}$, and adjacency matrix 
\begin{equation}
 \Gamma=\left(\begin{array}{cc}
                  0   & K^T\\
                  K & 0
                 \end{array}\right).
\end{equation}
\end{definition}
\noindent We remark that the state
\begin{equation}
 \ket{\psi_B'}=H^{\otimes m}\sum_{i=0}^{d^n-1} \braket{i}{\psi_A} \ket{K\vec{i}},
\end{equation}
where $\vec{i}$ is the vector of base-$D$ digits of $i$,
is produced up to by-product operators, if $K$ is injective, the qudits $A$ are initialized in an arbitrary state $\ket{\psi_A}$ instead of $\ket{+}^{\otimes m}$ and measured in $X$-basis~\cite{Beaudrap14}. This relates the considerations in terms of stabilizers to the quantum state produced in measurement-based quantum computations. See Appendix~\ref{app:inputoutputstate} for details.\\
As mentioned above, a linear network code is the decomposition of $F$ into a product of local encoding kernels, such that the rate constraints are fulfilled. A concatenation of the graph states corresponding to the local encoding kernels $K_v$ according to Definition~\ref{def:GK}, see FIG.~\ref{fig:graphextension}, translates the linear network code to a quantum network code. \begin{figure}[tp]%
\centering%
  \subfigure[A part of a linear network coding.]{\includegraphics[scale=1]{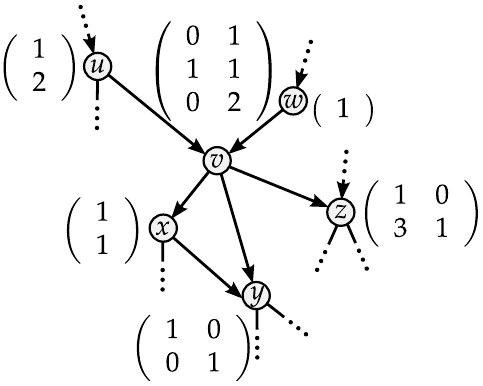}}\hfill%
  \subfigure[The same part in the quantum network coding.]{\includegraphics[scale=1]{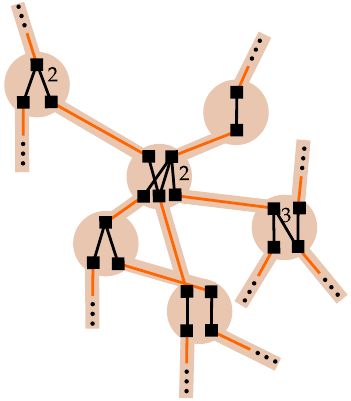}}%
 \caption{The construction of $G$ (b) from $\tilde{G}$ (a). The local subgraphs of (b) are given by Definition~\ref{def:GK} for the linear map of the local encoding kernels shown in (a). Orange (light grey) edges have weight $-1$ and black edges have weight one unless explicitly labelled otherwise.}\label{fig:graphextension}%
\end{figure}%
This is done in the following theorem. It generalizes the results of \cite{Beaudrap14}, which considers teleportation in the described network, to the distribution of general two-colorable graph states.
\begin{theorem}\label{thm:codeconstruction}
Let an $|A|$-dimensional linear network code on an acyclic network $G_{cl.}=(V_{cl.},E_{cl.})$ over base field $\mathds{F}_D$ be given by the local encoding kernels $K_v$ at each node $v\in V_{cl.}$. Let $A_i$, $i=1,2,...,|A|$, and $B_{j}$, $j=1,2,...,|B|$, denote qudits associated with each imaginary source and sink edge, respectively. Then there exists a measurement-based quantum network code that produces the graph state $\ket{G}$ on the qudits $V=A \cup B$, defined via the adjacency matrix of $G$,
\begin{equation}
 \Gamma=\left(\begin{array}{cc}
                    0 & F^T\\
                    F & 0
                   \end{array}\right),
\end{equation}
where the $j$-th row of $F$ is the global encoding kernel on the $j$-th imaginary sink edge.
\end{theorem}
The proof and more details on the construction of the pre-measurement graph state are given in Appendix~\ref{app:fromlinearcodes}. We remark that Theorem~\ref{thm:codeconstruction} contains the production of $|A|$ GHZ states as a special case, when the imaginary edges of the sink nodes are canonical basis vectors as defined by the imaginary edges of the source nodes. The construction of the quantum network code is illustrated for the example of a network of one source $p$ and six sinks that does not allow for a binary multicast, i.e. it is not possible to send two bits from $p$ to all sinks. But a ternary multicast is possible~\cite{Yeung08}. The corresponding QNC is shown in Fig.~\ref{fig:constructionexample}. \begin{figure}[tp]%
\centering%
  \subfigure[A ternary linear network coding.]{\includegraphics[scale=1]{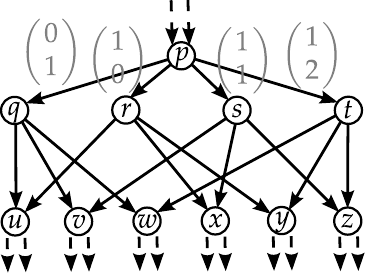}}\hfill %
  \subfigure[The corresponding quantum network coding.]{\includegraphics[scale=1]{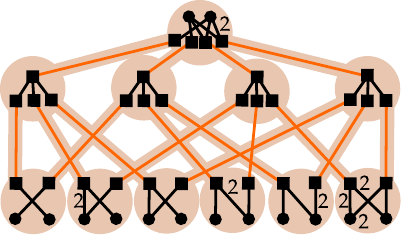}}%
 \caption{A linear network code (taken from Yeung, Fig. 19.5) and its corresponding quantum network code. Grey vectors are global encoding kernels. The network does not allow for a binary linear network code to achieve the task. The QNC produces two 7-qutrit GHZ states. Orange (light grey) edges have weight $-1$ and black edges have weight one unless explicitly labelled otherwise.}\label{fig:constructionexample}%
\end{figure}\\
\subsection{Equivalent network codes}
The pre-measurement graph state of Theorem~\ref{thm:codeconstruction} is not necessarily the simplest possible one. A network code can be simplified by using the graph state after the $X$-measurement on a qudit (see \cite{Hein04} for the case $D=2$) right from the start if it is compatible with the network constraints (rates of the channels). In many cases the state after $X$-measurements is given by the following theorem.
\begin{theorem}[Network simplification]\label{thm:simplification}
 Let $G=(V,E)$ with adjacency matrix $\Gamma$ be given. Denote the neighbourhood of $v\in V$ as $N_v=\{w\in V|\Gamma_{vw}\neq 0\}$.
 Let $a,b\in V$ be neighbours (i.e. $b\in N_a$), s.t. $N_a\cap N_b=\{\}$ and $\Gamma_{ab}$ and $D$ are co-prime. Via a measurement of $a$ in $X$-basis, the state $\ket{G}$ is projected onto $\ket{G'}$ (up to by-product operators) with adjacency matrix $\Gamma'$, which equals $\Gamma$ except for
\begin{align}
 \Gamma'_{cb}=&\frac{\Gamma_{ac}}{\Gamma_{ab}},\\
 \Gamma'_{cd}=&-\frac{\Gamma_{bd}\Gamma_{ac}}{\Gamma_{ab}}\\
 \text{and }\Gamma'_{ab}=&\Gamma'_{ac}=\Gamma'_{bd}=0,
\end{align}
where $c\in N_a\backslash \{b\}$ and $d\in N_b\backslash \{a\}$.
\end{theorem}
The theorem is proven by finding appropriate main stabilizers, see Appendix~\ref{app:proofofsimplification}. It can also be applied in the other direction. For example subdividing an edge by inserting a node on it leads to an equivalent quantum network code (this is a ``quantum repeater'').\\
Theorem~\ref{thm:simplification} can be applied to more general graph states than the ones of Theorem~\ref{thm:codeconstruction}, e.g. to networks where the graph of the distributed graph state is not two-colorable. Different codings, i.e. different pre-measurement graph states that are compatible with the network constraints, can lead to the same post-measurement graph state. Two equivalent codings are shown in FIG.~\ref{fig:equivalentcodings}. \begin{figure}[htbp] %
\subfigure[First variant.]{\includegraphics{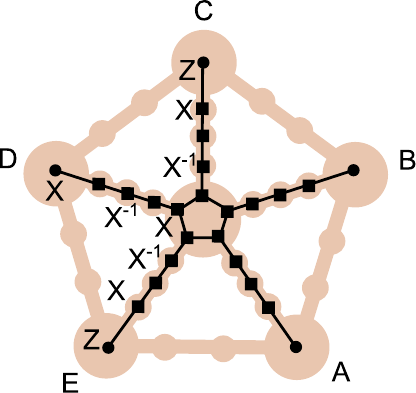}\label{fig:EquivalentCodingsa}} %
\subfigure[Second variant.]{\includegraphics{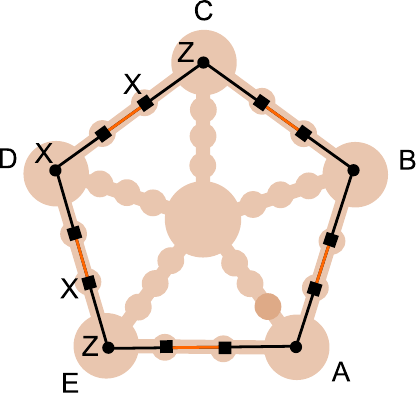}\label{fig:EquivalentCodingsb}} %
\subfigure[Post-measurement state.]{\includegraphics{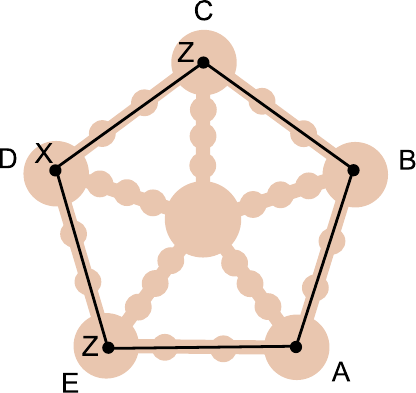}\label{fig:EquivalentCodingsc}} %
\caption{Different quantum network codes can lead to the same distributed graph state. Orange (light grey) edges have weight $-1$, black edges have weight one unless explicitly labelled otherwise. The brown background indicates the network sites and available links. The post-measurement state can be read from the main stabilizers like the one centred on $D$, which is shown in \subref{fig:EquivalentCodingsa} and \subref{fig:EquivalentCodingsb}. The final state is not local unitary equivalent to any two-colorable graph state.}\label{fig:equivalentcodings}
\end{figure}
\section{Robustness of quantum network coding}\label{sec:robustness}
An $[[n,k,d]]$ quantum error correction code encodes $k$ logical qudits into $n$ physical qudits with code distance $d$, i.e. it can correct up to $\lfloor \frac{d-1}{2} \rfloor$ errors at unknown positions or up to $d-1$ errors if the affected qudits are known. The measurement of a stabilizer operator detects an error that anti-commutes with it, because the measurement outcome becomes $-1$~\cite{Gottesman97,LidarBrunQEC13}.\\
In principle a graph state is an $[[n,k=0]]$ stabilizer quantum error correction code~\cite{Gottesman97,LidarBrunQEC13}, which can correct an arbitrary amount of errors. However, this would require to measure the stabilizer generators, which are not local and we consider this impossible in the network scenario. Also graph codes~\cite{Schlingemann01,Schlingemann01b}, i.e. quantum error correction codes where the codewords are elements of the graph state basis, which can encode data ($k>0$) require to measure non-local stabilizer operators in general.
But it is possible to detect a restricted set of errors from the data obtained in the $X$-measurements on the intermediate qudits in a QNC scheme. Namely any error that does not commute with an $X$-chain is detected~\cite{Wu15}, because the product of the measurement outcomes (to the appropriate power) on an $X$-chain is $1$ in case of no errors.\\
A tandem arrangement of weight one edges acts as an identity: The graph \raisebox{0.15\baselineskip}{\includegraphics{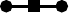}} is equivalent to \raisebox{0.15\baselineskip}{\includegraphics{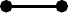}}, see Theorem~\ref{thm:simplification}, which obviously can be used for teleportation (see remark after Definition~\ref{def:GK}). The explicit calculation is done in Appendix~\ref{app:tandem}. We can thus form a three qudit repetition code as shown in FIG.~\ref{fig:repetitioncode}. %
\begin{figure}[tbp]%
 \subfigure[A three-qudit-repetition code.]{\label{fig:repetitioncode}\hspace{0.5cm}\includegraphics{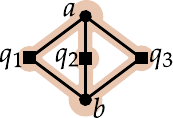}\hspace{0.5cm}}%
 \subfigure[The nine-qudit-Shor code.\label{fig:shorcode}]{\includegraphics{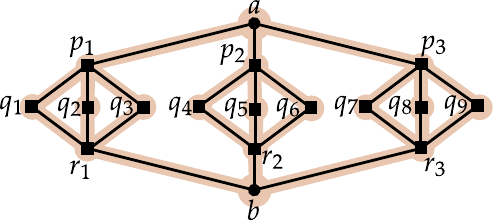}}%
\caption{Simple network variants of error correction codes. All edges have weight one.}\label{fig:simplecodes}%
\end{figure}%
We think of such a network as a distributed storage system: quantum information is sent to the three storage devices $q_1$, $q_2$, and $q_3$, where it is saved for some duration. Afterwards it is sent to $b$ where it is read out and/or processed. However, with some probability a memory gets corrupted or lost, e.g. due to some power outage, a malfunction or an adversary intervention. Suppose that an unnoticed $Z$ error occurs on one of the memories, say $q_1$, in the network of FIG.~\ref{fig:repetitioncode}. This error shifts (rotates) the outcome of the $X$-measurement on $q_1$. But, because $X_{q_1} X_{q_2}^{-1}$ and $X_{q_1} X_{q_3}^{-1}$ are $X$-chains of the graph state, the measurement outcomes of $q_1$, $q_2$ and $q_3$ need to be equal in case of no errors. By majority vote we decide that the error occurred on $q_1$ and we can correct it. There is also a slightly different way of understanding the error correction in a network. The network acts as the identity, because the main stabilizer operators $X_a X_b^{-1}$ and $Z_a X_{q_1} Z_b$ are projected onto $X_a X_b^{-1}$ and $Z_a Z_b$ by the $X$-measurements on the intermediate qudits. This Bell pair can be used to teleport a quantum state. Instead of choosing $Z_a X_{q_1} Z_b$ we could also choose $Z_a X_{q_2} Z_b$ or $Z_a X_{q_3} Z_b$, i.e. we only need to find one pair of error free main stabilizer operators to do teleportation. If the position of the error is known, which we will consider to be the case in the remainder, then two errors can be corrected/tolerated. We assume that an outage of a node is noticed.\\
Of course the presented three qudit repetition code is not able to correct $X$ errors on the storage devices. We therefore enlarge the code to the nine-qudit-Shor code, see FIG.~\ref{fig:shorcode}. 
The $X$-chain group on the intermediate qudits is generated by the operators
\begin{equation}
 \begin{aligned}
  s_1=&X_{q_1} X_{q_3}^{-1},\\
  s_2=&X_{q_2} X_{q_3}^{-1},\\
  s_3=&X_{q_4} X_{q_6}^{-1},\\
  s_4=&X_{q_5} X_{q_6}^{-1},\\
  s_5=&X_{q_7} X_{q_9}^{-1},\\
  s_6=&X_{q_8} X_{q_9}^{-1},\\
  s_7=&X_{p_1} X_{r_1} X_{p_3}^{-1} X_{r_3}^{-1},\\
\text{and }  s_8=&X_{p_2} X_{r_2} X_{p_3}^{-1} X_{r_3}^{-1},
 \end{aligned}
\end{equation}
corresponding to the generators of the stabilizer of the code space of the nine-qudit-Shor code. A possible choice of the two main stabilizer operators that are projected onto the stabilizer generators of a Bell pair are
\begin{align}
 S_1=&X_a X_{q_1}^{-1} X_{q_4}^{-1} X_{q_6}^{-1} X_b\\
 S_2=&Z_a X_{p_1} X_{r_1}^{-1} Z_b^{-1},
\end{align}
while again multiplication of $S_i$ with any $X$-chain leads to another possible choice for $S_i$. As before, the error $Z_{q_1}$ shifts the outcome of the $X$-measurement outcome on $q_1$. An $X$-error on $q_1$ does not affect the outcome at $q_1$, but it propagates to an $Z$-error on $r_1$ via the $C_Z$-gate. Again we assume that the positions of losses are known. Then $S_1$ requires that there is at least one error-free qudit in each group ($\{q_1,q_2,q_3\}$, $\{q_4,q_5,q_6\}$ and $\{q_7,q_8,q_9\}$), while $S_2$ requires, that there is at least one group with no loss at all. The nine-qudit-Shor code can be generalized to an arbitrary amount of blocks and an arbitrary number of qudits per block, which is sometimes referred to as quantum parity code~\cite{Ralph05}.\\
The link between the network of Fig.~\ref{fig:shorcode} and the Nine-qubit-Shor code, and an analogue correspondence for all network error correction codes, is made formal in the following definition.
\begin{definition}[Network error correction code]\label{def:networkerrorcorrection}
A graph state $\ket{G}$ associated with $G=(V,E)$ with vertices $V=A\cup C \cup B$ (sources, intermediate nodes and sinks) and adjacency matrix $\Gamma$ defines an $[[n,k,d]]$ quantum network code w.r.t. $Q=\{q_1,q_2,...,q_n\}\subset C$ (the memories), if
\begin{enumerate}
 \item for $i\in\{1,2,...,k\}$ the operators $X_{a_i} Z_{b_i}$ and $Z_{a_i} X_{b_i}$ stabilize the state after $X$-measurements on $C$ and application of by-product operators (on $A$ and $B$)
 \item and $d$ is the smallest weight of an error acting on $Q\subset V$ that commutes (after error propagation) with all $X$-chains on $C$.
 \item Each main stabilizer between $A$ and $B$ acts on $Q$ or $Out(Q)$ non-trivially.
\end{enumerate}
\end{definition}
In the following we motivate the previous definition by mapping the network code to a stabilizer error correction code. An $X$-chain of a graph state is a stabilizer operator which does not contain $Z$-operators~\cite{Wu15}. $X$-chains form a group with the usual multiplication. The $X$-measurements on $C$ determine the measurement outcomes of all $X$-chains simultaneously. Let $s_i$ denote the generators of the $X$-chain group on $G$ restricted to $C$ (i.e. they act trivially on $A$ and $B$). This group is mapped to the stabilizer group of an $[[n,k,d]]$ stabilizer error correction code generated by
\begin{equation}
 \tilde{s}_i = \prod_j X_j^{\alpha_{ij}} \prod_{h=1}^{|Out(q_j)|} Z_{q_j}^{\beta_{ijh} \Gamma_{q_j,Out(q_j)_j}}
\end{equation}
with $\alpha_{ij}$ and $\beta_{ijh}$ being the power of $X_{q_j}$ and $X_{Out(q_j)_h}$ in $s_i$, respectively. That is, $\tilde{s}_i$ contains an $X$ operator for each $X$ operator on $Q$ and a $Z$ operator for each $X$ operator on a successor of $Q$, respectively, in $s_i$. The $C_Z$ gates corresponding to the edges between the qudits in $Q$ and their successors $Out(Q)$ propagate $X$-errors on $Q$ to $Z$-errors on $Out(Q)$. Thus after error propagation an error on $Q$ commutes with $\tilde{s}_j$ if and only if the same error in the stabilizer code commutes with $s_j$. An error is undetected if and only if it commutes with all stabilizer generators. Therefore the network code with generators $s_j$ and the stabilizer code with generators $\tilde{s}_j$ have the same error correction capabilities.\\
The mapping of Definition~\ref{def:networkerrorcorrection} can be used to analyse the robustness of a QNC w.r.t. errors on $Q$. We do this for the two networks shown in FIG.~\ref{fig:diversitycoding}. %
\begin{figure}[htbp]%
 \subfigure[A QNC corresponding to the linear code of FIG. 18.5 in \cite{Yeung08}.]{\centering \hspace{1cm}\includegraphics{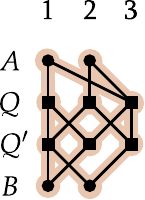}\hspace{1cm}}%
 \subfigure[A $\lbrack\lbrack 4,2,2 \rbrack\rbrack$ QNC.\label{fig:schubertcode}]{\centering \hspace{1cm}\includegraphics{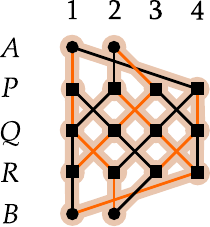}\hspace{1cm}}%
 \caption{Quantum network coding can be used for diversity coding~\cite{Yeung08}: Errors on some of the memories (Q) can be corrected using the data of the $X$-measurements on all qudits except at A and B.}\label{fig:diversitycoding}%
\end{figure}%
The first network is a QNC analogue of the diversity coding scheme of FIG. 18.5 in \cite{Yeung08}. In this example, a possible choice of main stabilizer operators reads
\begin{align}
S_{A_1}=& X_{A_1} X_{Q_2'}^{-1} Z_{B_1}^{-1},\\
S_{A_2}=& X_{A_2} X_{Q_3'}^{-1} Z_{B_2}^{-1},\\
S_{B_1}=& X_{B_1} X_{Q_1}^{-1} Z_{A_1}^{-1},\\
\text{and }S_{B_2}=& X_{B_2} X_{Q_2}^{-1} Z_{A_2}^{-1},
\end{align}
i.e. two Bell pairs shared by A and B are produced.\\
One easily verifies that the $X$-chain group on $C=Q\cup Q'$ is $\langle X_{q_1} X_{q_2} X_{q_3}^{-1}, X_{q_1'}^{-1}X_{q_2'}X_{q_3'}\rangle$. These generators correspond to $\tilde{s}_1=X^{-1} X X$ and $\tilde{s}_2=\mathds{1}\mathds{1}Z^2$. While this code can detect a single $Z$ error on any memory, it cannot detect an $X$ error, i.e. it can only correct one type of errors. Consequently it has $d=1$. One might expect this result for a classical code.\\
The generators of the $X$-chain group on $C= P\cup Q \cup R$ in the expanded network in FIG.~\ref{fig:schubertcode} and the corresponding codespace stabilizer operators are
\begin{equation}
 \begin{aligned}
  s_1=&X_{q_1}X_{q_2}X_{q_3}X_{q_4}, &  \tilde{s}_1=&XXXX,\\
  s_2=&X_{p_1}X_{p_4}X_{r_1}X_{r_4}, &  \tilde{s}_2=&ZZ^{-1}Z^{-1}Z,\\
  s_3=&X_{p_2}X_{p_3}X_{r_2}X_{r_3}, &  \tilde{s}_3=&Z^{-1} Z Z Z^{-1},\\
  s_4=&X_{p_1}X_{p_2}X_{p_3}X_{p_4}, \text{ and} &  \tilde{s}_4=&\mathds{1}\mathds{1}\mathds{1}\mathds{1}.
 \end{aligned}
\end{equation}
From this one readily obtains that the QNC of FIG.~\ref{fig:schubertcode} is a $[[4,2,2]]$ code, i.e. it encodes two qudits into four qudits and it can tolerate a single erasure in $Q$.\\
A construction similar to the previous ones can be used to obtain a quantum network code from any stabilizer quantum error correction code.
\begin{theorem}\label{thm:stabcodes}
Given an $[[n,k,d]]$ stabilizer quantum error correction code $\mathcal{C}$ with parity check matrices $H_X$ and $H_Z$ in prime dimension $D$. Let the $i$-th logical $X$-operator be given by the numbers $x_{ij}\in \mathds{F}_D$, s.t.
\begin{equation}
\bar{X}_i=\prod_{j=1}^n X_j^{x_{ij}}.
\end{equation}
Denote $V=A \cup A' \cup B \cup B' \cup S \cup Q \cup S'$, with $A=\{a_1,a_2,...,a_k\}$, $A'=\{a_1',a_2',...,a_k'\}$, $B=\{b_1,b_2,...,b_k\}$, $B'=\{b_1',b_2',...,b_k'\}$, $S=\{s_1,s_2,...,s_{n-k}\}$, $S'=\{s_1',s_2',...,s_{n-k}'\}$, and $Q=\{q_1,q_2,...,q_n\}$.
The network with vertices $V$ and adjacency matrix
\begin{equation}
\Gamma=\left(\begin{array}{ccccccc}
           & -\Gamma_1 & & & & & \\
-\Gamma_1 &  & & & & \Gamma_X & \\
& & & \Gamma_1 & & & \\
& & \Gamma_1 & & &  & -\Gamma_X\\ 
& &  & & -\Gamma_S & H_Z &\\ 
& \Gamma_X &  &  & H_Z^T &  & -H_Z^T\\ 
&  &  & -\Gamma_X &  & -H_Z & \Gamma_S
\end{array}\right)
\end{equation}
with 
\begin{align}
(\Gamma_X)_{ij}=x_{ij},\\
(\Gamma_1)_{ij}=&\delta_{ij} \mathrm{wt}(\bar{X}_i),\\
\text{and }(\Gamma_S)_{ij}=&\sum_l (H_X)_{il} (H_Z)_{jl}
\end{align}
implements an $[[n,k,d]]$ QNC w.r.t. errors on $Q$. Here $\mathrm{wt}(\bar{X}_i)$ denotes the number of qudits on which the $i$-th logical $X$ operator acts non-trivially. The network produces $k$ Bell pairs and thus allows to teleport $k$ qudits from $A$ to $B$. 
\end{theorem}
The proof is done by identifying main stabilizer operators and is given in Appendix~\ref{app:stabcodes}. The idea of the construction of Theorem~\ref{thm:stabcodes} is illustrated in FIG.~\ref{fig:stabcodes} for the 7-qubit Steane code, the five-qubit code and a [[12,2,3]] low-density parity check code taken from~\cite{LidarBrunQEC13}. The error correction in the network variant is done analogously to the stabilizer code. For example in FIG.~\ref{fig:fivequbits}, the $X$-chains
\begin{equation}
\begin{aligned}
s_1=& X_{q_1} X_{s_1}X_{s_1'} X_{q_4},\\
s_2=& X_{q_2} X_{s_2}X_{s_2'} X_{q_5},\\
s_3=& X_{q_1} X_{q_3} X_{s_3}X_{s_3'},\\
\text{and } s_4=& X_{q_2} X_{q_4} X_{s_4}X_{s_4'}
\end{aligned}
\end{equation}
correspond to the generators
\begin{equation}
\begin{aligned}
\tilde{s}_1=& X\otimes Z\otimes Z\otimes X \otimes\1,\\
\tilde{s}_2=& \1 \otimes X\otimes Z\otimes Z\otimes X,\\
\tilde{s}_3=& X\otimes \1 \otimes X \otimes Z \otimes Z,\\
\text{and } \tilde{s}_4=& Z\otimes X\otimes \1\otimes X\otimes Z
\end{aligned}
\end{equation}
in the five-qubit code.
 \begin{figure}[tbp] %
 \subfigure[{7-qubit Steane code ([[7,1,3]])~\cite{LidarBrunQEC13}.}]{\includegraphics{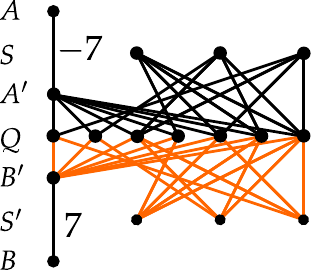}}\hspace{1cm}% 
 \subfigure[{5-qubit code ([[5,1,3]])~\cite{LidarBrunQEC13}.}]{\includegraphics{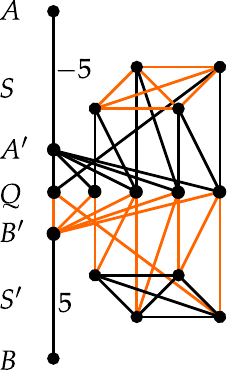}\label{fig:fivequbits}}\\ %
 \subfigure[{A $[[12,2,3]]$ QNC~\cite[p.~298]{LidarBrunQEC13}.}]{\includegraphics{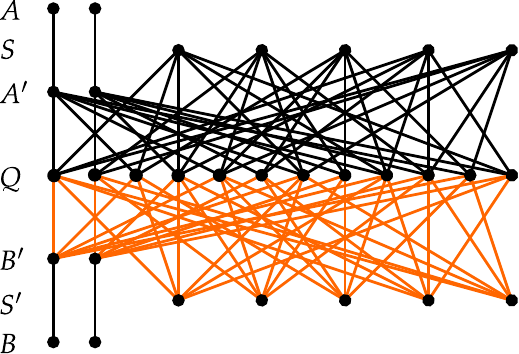}}
 \caption{Three examples for Theorem~\ref{thm:stabcodes}. By measuring all qudits except at A and B in $X$-basis, $k$ Bell-pairs shared by A and B are obtained (up to by-product operators). The distribution scheme can tolerate two errors (i.e. node outages) on the $n$ ``memories''-qudits $Q$. \label{fig:stabcodes}}%
\end{figure}\\%
Please note that the QNC described in Theorem~\ref{thm:stabcodes} is not necessarily the simplest one implementing the stabilizer error correction code in a network (see also Theorem~\ref{thm:simplification}). The QNC corresponding to the nine-qubit-Shor code obtained from Theorem~\ref{thm:stabcodes} differs from the example shown in Figure~\ref{fig:shorcode}.
\FloatBarrier
\section{Conclusion and outlook}\label{sec:conclusion}
In this paper we described how graph state repeater with vertex degree larger than two, which we call quantum router, can employ network coding. It is well known that this measurement based quantum computation at intermediate network sites can increase the throughput of a network.\\
We used the framework of graph state repeaters to generalize the construction of QNCs from linear codes of \cite{deBeaudrap14} to the distribution of general graph states associated with two-colorable graphs. Furthermore we analysed the robustness of QNC, i.e. its error correction capabilities, by mapping them to stabilizer codes.\\
The formalism developed in this work can be applied to investigate further interesting questions in this context. We plan to apply the techniques developed to simplify the quantum network code to the simplification of gates in measurement-based quantum computation, e.g. the SWAP gate that corresponds to the butterfly network.\\
It will be interesting to consider network coding in random graphs and compare it to classical random network coding, which proved to be very efficient, e.g. if the network layout is unknown.\\
A further related project is the generalization of the presented scheme to the distribution of graph states associated with non-two-colorable graphs.\\
The advantage of the distribution of multipartite entanglement via routers with network coding for quantum key distribution, as measured by the secret key rate, compared to bipartite QKD protocols, deserves further quantitative investigations.

\bibliographystyle{apsrev4-1} 
%\bibliography{citations}

\appendix
\section{By-product operators on graph states}
\label{app:byproduct}
The $X$-measurements on the intermediate qudits introduce phase factors in the stabilizer operators according to the measurement outcomes. The aim of the discussed protocols is, however, to produce the same final state independent of the random measurement outcomes at intermediate sites. It is therefore necessary to correct for these unwanted factors, either by physically applying unitary operations on the unmeasured qudits or by performing these corrections on the classical data after measurements on the final state (if possible). These corrections are called by-product operators.\\
The phase factors introduced in the main stabilizer operators by $X$-measurements on intermediate qudits can always be corrected. Suppose, for example, the stabilizer operator $S_A=\eulere^{x \frac{2\pi \ramuno}{D}} X_A Z_B$ that accumulated the unwanted phase $\eulere^{x \frac{2\pi \ramuno}{D}}$ with $x\in \mathds{N}$. Because $Z^n X^m=\eulere^{mn \frac{2\pi\ramuno}{D}} X^m Z^n $ for any $m,n\in\mathds{N}$, the by-product operator $Z_A^{-x}$ can be applied to cancel out the phase $\eulere^{x\frac{2\pi \ramuno}{D}}$ that depends on the measurement outcomes.\\
Additionally, local basis changes might be necessary to obtain a graph state in the standard form. Consider, for example, the graph state \raisebox{0.15\baselineskip}{\includegraphics{tandemidentitya}} stabilized by $X_1 X_3^{-1}$ and $Z_1 X_2 Z_3$. The $X_2$-measurement leads to $X_1 X_3^{-1}$ and $Z_1 Z_3$ (already phase-corrected). The local basis change $H_3^\dagger$ brings these stabilizer operators into standard form $X_1 Z_3$ and $Z_1 X_3$.
\section{Measurement-based computation with a two-colorable graph state}\label{app:inputoutputstate}
Consider the qudits $A$ to be in a computational basis state $\ket{\vec{j}}$.
The gate from $A_x$ to $B_y$ acts as $Z^{j^{(x)} K_{yx}}$ on $B_y$.
The state after application of the $C_Z$ gates is
\begin{align*}
\ket{\vec{j}}\bigotimes_y \prod_x Z^{j^{(x)} K_{yx}} \ket{+}_y
=& \ket{\vec{j}}\bigotimes_y Z^{\sum_x K_{yx} j^{(x)}}\ket{+}_y.\\
=& \ket{\vec{j}}\bigotimes_y \ket{\vec{j}} H^2 Z^{\sum_x K_{yx} j^{(x)}} H^2 \ket{+}_y\\
=&\ket{\vec{j}}\bigotimes_y H X^{\sum_x K_{yx} j^{(x)}} \ket{0}_y\\
=&\ket{\vec{j}}\bigotimes_y H \ket{ (K\vec{j})_y }\\
=&\ket{\vec{j}} \otimes H^{\otimes m} \ket{K\vec{j}}.\\
\end{align*}
This carries over to the superposition $\ket{\psi_A}=\sum_{i=0}^{d^n-1} \braket{\vec{i}}{\psi_A} \ket{\vec{i}}$, i.e.
\begin{equation}
 \sum_j \braket{\vec{j}}{\psi_A} \ket{\vec{j}} \otimes H^{\otimes m} \ket{K\vec{j}}.
\end{equation}
Now we express $\ket{\vec{j}}$ in the $X$-basis as
\begin{equation}
 \ket{j}=\frac{1}{\sqrt{d}}\sum_{y} \eulere^{-\frac{2\pi \ramuno}{d} y j} H\ket{y}
\end{equation}
i.e.
\begin{equation}
 \ket{\vec{j}}=\frac{1}{\sqrt{d^n}}H^{\otimes n}\sum_{y=0}^{d^n-1} \eulere^{-\frac{2\pi \ramuno}{d} \vec{y}\cdot \vec{j}} \ket{\vec{y}}
\end{equation}
and get
\begin{align}
 &\sum_j \braket{\vec{j}}{\psi_A} \ket{\vec{j}} \otimes H^{\otimes m} \ket{K\vec{j}}\\
 =&\sum_j \braket{\vec{j}}{\psi_A} \frac{1}{\sqrt{d^n}}H^{\otimes n}\sum_{y} \eulere^{-\frac{2\pi \ramuno}{d}\vec{y}\cdot \vec{j}} \ket{\vec{y}} \otimes H^{\otimes m} \ket{K\vec{j}}\\
 =&\frac{1}{\sqrt{d^n}}  \sum_{y}  H^{\otimes n} \ket{\vec{y}} \otimes H^{\otimes m} \sum_j \braket{\vec{j}}{\psi_A}  \eulere^{-\frac{2\pi \ramuno}{d} \vec{y}\cdot \vec{j}} \ket{K\vec{j}}
\end{align}
The by-product operator
\begin{equation}
\sum_{\vec{i}=K\vec{j}} \eulere^{\frac{2 \pi \ramuno}{d} \vec{y}\cdot\vec{j}} H^{\otimes m}\ket{\vec{i}}\bra{\vec{i}} H^{\otimes m}
\end{equation}
can be constructed from $X$-operations, where one uses that $K$ is injective. See also \cite{Beaudrap14}.
\section{Construction for linear network codes}\label{app:fromlinearcodes}
Without loss of generality we assume that each packet received by a node is required for the calculation of some outgoing packets. If this was not the case, one could simplify the network coding by not sending that packet.
We start by constructing the graph $G'$ that defines the pre-measurement graph state $\ket{G'}$.
The construction is the same as described in \cite{Beaudrap14}, i.e. every node is replaced by the graph corresponding to the local encoding kernel. More explicitly, for each network node $v\in V_{cl.}$ we insert $|\In(v)|+|\Out(v)|$ associated vertices into $V'$, which we denote $V'(v)=\{v_{in,1},...,v_{in,|\In(v)|},v_{out,1},...,v_{out,|\Out(v)|}\}$.  The local adjacency matrix is given by
\begin{equation}
 \Gamma_v=\left(\begin{array}{cc}
                                0 & K_v^T\\
                                K_v & 0
                               \end{array}\right),
\end{equation}
see also Definition~\ref{def:GK}. Note that our assumption on the linear network coding implies that no row of $\Gamma_v$ is zero. For each edge $e=(a,b)\in E$ we add one undirected edge $(a_{out,i},b_{in,j})$ of weight $D-1$ to $E'$, where $i$ and $j$ fulfil $\Out(a)_i=b$ and $\In(b)_j=a$, respectively. See also FIG.~\ref{fig:graphextension}. The graph state $\ket{G'}$ can be distributed within the rate constraints of the edges $E$, because for each edge in $E$ there exists a single edge $e'=(a_{out,i},b_{in,j})\in E'$. For large-scale quantum networks, graph repeaters can be employed. Remember that the additional classical communication required for a repeater network is considered to be free in the present context.\\
We make use of a converse network coding on the graph $G^T=(V_{cl.},E_{cl.}^T)$ with $E_{cl.}^T=\{(b,a)|(a,b)\in E_{cl.}\}$ given by the local coding kernels $K_v^T$ and global encoding kernels $\vec{g}_e$. Note that this changes the dimension $|A|$ of the network to $|B|$. We denote the source qudits $A_i$ ($i=1,2,...,|A|$) and the sink qudits $B_{\tilde{i}}$ ($\tilde{i}=1,2,...,|B|$). The canonical basis vectors associated with the symbol send by $B_{\tilde{i}}$ in the converse coding are denoted by $\vec{e}_{\tilde{i}}$. The matrix $F$, that contains the global encoding kernels of the sink edges and represents the net effect of the network, transposes when going to the converse code.\\
To project the state $\ket{G'}$ onto the desired state $\ket{G}$, all qudits except the source (in-) and the sink (out-) qudits, i.e. the intermediate qudits $C$, are measured in $X$-basis. Note that $X$ can be measured even though it is not Hermitian and from this also the measurement outcomes of any powers of $X$ are known. The main-stabilizer operators of the QNC are constructed from the global encoding kernels. They read
\begin{align}
 S_{A_i} =& g_{A_i} \prod_{e=([\In(b)]_j,b)\in E_{cl.}} g_{b_{in,j}}^{(\vec{f}_e)_i} \label{eq:SAi1}\\
         =& X_{A_i} \prod_{\tilde{i}} Z_{B_{\tilde{i}}}^{F_{\tilde{i}i}} \prod_{e=([\In(b)]_j,b)\in E_{cl.}} X_{b_{in,j}}^{(\vec{f}_e)_i} \label{eq:SAi2}
\intertext{and}
 S_{B_{\tilde{i}}} =& g_{B_{\tilde{i}}} \prod_{e=(a,[\Out(a)]_j)\in E_{cl.}} g_{a_{out,j}}^{(\vec{g}_e)_{\tilde{i}}}\label{eq:SBim1}\\
                   =& X_{B_{\tilde{i}}} \prod_{i} Z_{A_i}^{F_{\tilde{i}i}} \prod_{e=(a,[\Out(a)]_j)\in E_{cl.}} X_{a_{out,j}}^{(\vec{g}_e)_{\tilde{i}}}  .\label{eq:SBim2}
\end{align}
\begin{proof}
We start by proving
\begin{align}
 S_{A_i} =& g_{A_i} \prod_{e=([\In(b)]_j,b)\in E_{cl.}} g_{b_{in,j}}^{(\vec{f}_e)_i}\label{eq:SAi1app}\\
         =& X_{A_i}\prod_{\tilde{i}} Z_{B_{\tilde{i}}}^{F_{\tilde{i}i}} \prod_{e=([\In(b)]_j,b)\in E_{cl.}} X_{b_{in,j}}^{(\vec{f}_e)_i}\label{eq:SAi2app}
\end{align}
for a fixed but arbitrary $i\in\{1,2,...,|A|\}$. While the $X$-operators in Eq.~(\ref{eq:SAi2app}) are clear, we need to show that exactly the given $Z$-operators appear in the stabilizer. Furthermore it is clear, that no $Z$ operator can occur on a vertex $v_{in,j}\in V'$. Therefore consider any vertex $c=v_{out,j}\in V'$. Let $e$ denote the $j$-th edge going out of $v$. %
\begin{figure}[htbp]%
\centering\includegraphics[width=\linewidth]{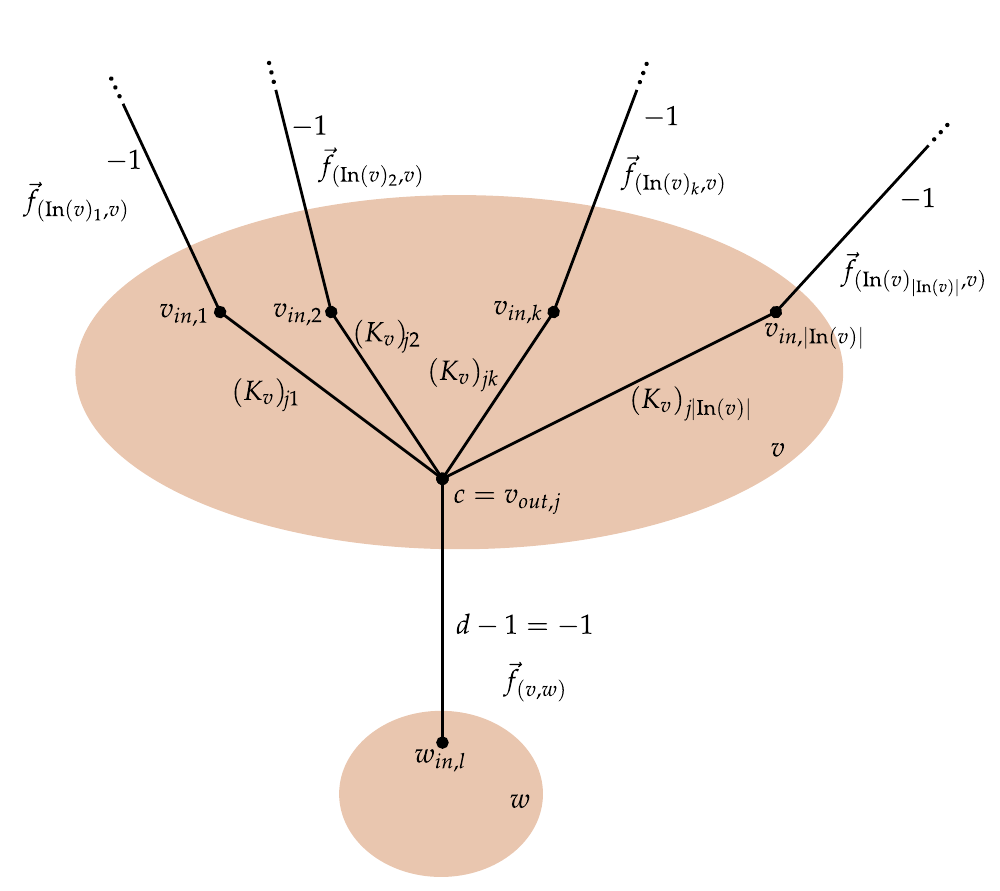}%
\caption{Situation considered in the derivation of the main stabilizers.}%
\end{figure}%
In the following we note that the calculation of the exponent of $Z_c$ in Eq.~(\ref{eq:SAi2app}) is analogous to the calculation of $(\vec{f}_e)_i$ in the classical linear code:\\
Remember that the global encoding kernel of the $j$-th outgoing edge is calculated from the global encoding kernels of the incoming edges and the local encoding kernel via
\begin{equation}
(\vec{f}_e)_i = \sum_{k=1}^{|\In(v)|} (K_v)_{jk} (\vec{f}_{(\In(v)_k,v)})_i \mod d.
\end{equation}
This calculation is done modulo $d$. Now in the QNC the operator on the $k$-th parent of $c$ in Eq.~(\ref{eq:SAi1app}) is $g_{v_{in,k}}^{(\vec{f}_{(\In(v)_k,v)})_i}$. Because the edges in $E'$ to $c$ are weighted according to $K_v^T$, the exponent $h$ of $Z_b^h$ is
\begin{equation}
h=\sum_{k=1}^{|\In(v)|} (K_v)_{jk} (\vec{f}_{(\In(v)_k,v)})_i \mod d=(\vec{f}_e)_i.
\end{equation}
We used that $Z^d=\mathds{1}$. So the stabilizer generators centred on the parents of $c$ in Eq.~(\ref{eq:SAi1app}), $\prod_{k=1}^{|\In(v)|} g_{v_{in,k}}^{(\vec{f}_{(\In(v)_k,v)})_i}$, contribute $Z_c^{(\vec{f}_e)_i}$ on this qudit $c$.
We distinguish two cases.
\begin{enumerate}
\item $c$ is the $\tilde{i}$-th sink qudit. Then the outgoing edge $e$ is imaginary, i.e. $(\vec{f}_e)_i=F_{\tilde{i}i}$ and the operator on $c$ reads $Z_{B_{\tilde{i}}}^{F_{\tilde{i}i}}$.
\item $c$ is an intermediate qudit. Let $w$ be the child of $v$ in $G$ and let $e=(v,w)\in E$. Furthermore let $w_{in,l}$ be the qudit associated with $w$ that is connected to $c$, i.e. $(c,w_{in,l})\in E'$. Then the operator $Z_c^{(\vec{f}_e)_i}$ obtained from the parents of $v$ is cancelled by the action of the operator $g_{w_{in,l}}^{(\vec{f}_e)_i}$. This is due to the edge weight $D-1=-1$ of $(c,\Out(v)_j)\in E'$.
\end{enumerate}
This finishes the proof of Eq.~(\ref{eq:SAi2app}). We continue by showing that
\begin{align}
 S_{B_{\tilde{i}}} =& g_{B_{\tilde{i}}} \prod_{e=(a,[\Out(a)]_j)\in E_{cl.}} g_{a_{out,j}}^{(\vec{g}_e)_{\tilde{i}}}\label{eq:SBim1app}\\
         =& X_{B_{\tilde{i}}} \prod_i Z_{A_i}^{F_{\tilde{i}i}} \prod_{e=(a,[\Out(a)]_j)\in E_{cl.}} X_{a_{out,j}}^{(\vec{g}_e)_{\tilde{i}}}  .\label{eq:SBim2app}
\end{align}
In the converse network coding on $G_{cl.}^T$, source and sinks exchange their role. The same reasoning as above can be applied to prove Eq.~(\ref{eq:SBim2app}).
\end{proof}
Denote the measurement outcome of the intermediate node $w$ as $\eulere^{x_w 2\pi \ramuno/D}$. The operator $X_w$ in Eqs.~(\ref{eq:SAi2}) and (\ref{eq:SBim2}) is replaced by $\eulere^{x_w 2\pi \ramuno/D}$. By this $S_v$ accumulates a phase $\eulere^{\xi_v 2\pi \ramuno/D}$ with
\begin{equation}
 \xi_v=\left\{\begin{aligned}
\sum_{e=([\In(b)]_j,b)\in E}  &x_{b_{in,j}} (\vec{f}_e)_i & & \text{if }v=A_i\\
\sum_{e=(a,[\Out(a)]_j)\in E} &x_{a_{out,j}} (\vec{g}_e)_i & & \text{if }v=B_{\tilde{i}}
       \end{aligned}\right.
\end{equation}
and a possible choice of the corresponding by-product operator is $Z_v^{(D-1)\xi_v}$. One sees that $\ket{G}$ has been produced by comparing its stabilizer generators with the main stabilizers.\\

\section{Proof of Theorem~\ref{thm:simplification}}\label{app:proofofsimplification}
\setcounter{theorem}{1}
\begin{theorem}[Network simplification]\label{thm:simplificationapp}
 Let $G=(V,E)$ with adjacency matrix $\Gamma$ be given. Denote the neighbourhood of $v\in V$ as $N_v=\{w\in V|\Gamma_{vw}\neq 0\}$.
 Let $a,b\in V$ be neighbours (i.e. $b\in N_a$), s.t. $N_a\cap N_b=\{\}$ and $\Gamma_{ab}$ and $D$ are co-prime. Via a measurement of $a$ in $X$-basis, the state $\ket{G}$ is projected onto $\ket{G'}$ (up to by-product operators) with adjacency matrix $\Gamma'$, which equals $\Gamma$ except for
\begin{align}
 \Gamma'_{cb}=&\frac{\Gamma_{ac}}{\Gamma_{ab}},\\
 \Gamma'_{cd}=&-\frac{\Gamma_{bd}\Gamma_{ac}}{\Gamma_{ab}}\\
 \text{and }\Gamma'_{ab}=&\Gamma'_{ac}=\Gamma'_{bd}=0,
\end{align}
where $c\in N_a\backslash \{b\}$ and $d\in N_b\backslash \{a\}$.
\end{theorem}
\begin{proof}
\begin{figure}
\subfigure[pre-measurement]{\includegraphics{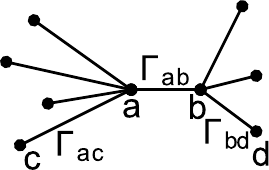}}\hspace{1cm}
\subfigure[post-measurement]{\includegraphics{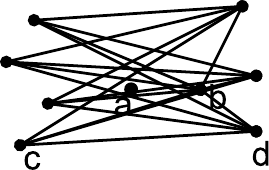}}
\caption{Sketch of the graph in Theorem~\ref{thm:simplification} with vertices a, b, c, and d.}\label{fig:theorem2sketch}
\end{figure}
The notation is illustrated in FIG.~\ref{fig:theorem2sketch}. The pre-measurement graph state $\ket{G}$ is stabilized by the main stabilizer operators
\begin{align}
 g_a =&X_a\prod_{c\in N_a\backslash\{b\}} Z_c^{\Gamma_{ac}} Z_b^{\Gamma_{ab}},\\
 \intertext{$\forall c\in N_a\backslash\{b\}:$}
 g_c g_b^{-\frac{\Gamma_{ac}}{\Gamma_{ab}}} =&X_c\prod_{i\in N_c\backslash\{a\}} Z_i^{\Gamma_{ic}} X_b^{-\frac{\Gamma_{ac}}{\Gamma_{ab}}} \prod_{d\in N_b\backslash\{a\}} Z_d^{-\frac{\Gamma_{bd}\Gamma_{ac}}{\Gamma_{ab}}}\\
 \intertext{and $\forall d\in N_b\backslash\{a\}:$}
 g_d g_a^{-\frac{\Gamma_{bd}}{\Gamma_{ab}}} =& X_d\prod_{i\in N_d\backslash\{b\}} Z_i^{\Gamma_{id}} X_a^{-\frac{\Gamma_{bd}}{\Gamma_{ab}}} \prod_{c\in N_a\backslash\{b\}} Z_c^{-\frac{\Gamma_{bd}\Gamma_{ac}}{\Gamma_{ab}}}.
\end{align}
The calculation of the exponents is done modulo $D$. The exponents of $g_b$ and $g_a$ are chosen such that the $Z$ operators on $a$ and $b$, respectively, cancel out. $\Gamma_{ab}$ can be inverted if and only if $\Gamma_{ab}$ and $D$ are co-prime. Now apply $H_b^\dagger$ on the state (i.e. $H_b^\dagger (\cdot) H_b$ on a stabilizer operator) and replace $X_{a}$ by $1$ (i.e. do measurement and apply by-product operator) to obtain
\begin{align}
 g_b'=&X_b \prod_{c\in N_a\backslash\{b\}} Z_c^{\frac{\Gamma_{ac}}{\Gamma_{ab}}}, \\
 g_c'=&X_c Z_b^{\frac{\Gamma_{ac}}{\Gamma_{ab}}} \prod_{i\in N_c\backslash\{a\}} Z_i^{\Gamma_{ic}} \prod_{d\in N_b\backslash\{a\}} Z_d^{-\frac{\Gamma_{bd}\Gamma_{ac}}{\Gamma_{ab}}}\\
 \text{and }g_d'=&X_d\prod_{i\in N_d\backslash\{b\}} Z_i^{\Gamma_{id}} \prod_{c\in N_a\backslash\{b\}} Z_c^{-\frac{\Gamma_{bd}\Gamma_{ac}}{\Gamma_{ab}}}.
\end{align}
We used that $H^\dagger X H = Z^{-1}$ and $H^\dagger Z H=X$. From the exponents of the $Z$-operator one can read the weights of the post-measurement state.
\end{proof}
\section{Proof of Theorem~\ref{thm:stabcodes}}\label{app:stabcodes}
\begin{theorem}\label{thm:stabcodesapp}
	Given an $[[n,k,d]]$ stabilizer quantum error correction code $\mathcal{C}$ with parity check matrices $H_X$ and $H_Z$ in prime dimension $D$. Let the $i$-th logical $X$-operator be given by the numbers $x_{ij}\in \mathds{F}_D$, s.t.
	\begin{equation}
	\bar{X}_i=\prod_{j=1}^n X_j^{x_{ij}}.
	\end{equation}
	Denote $V=A \cup A' \cup B \cup B' \cup S \cup Q \cup S'$, with $A=\{a_1,a_2,...,a_k\}$, $A'=\{a_1',a_2',...,a_k'\}$, $B=\{b_1,b_2,...,b_k\}$, $B'=\{b_1',b_2',...,b_k'\}$, $S=\{s_1,s_2,...,s_{n-k}\}$, $S'=\{s_1',s_2',...,s_{n-k}'\}$, and $Q=\{q_1,q_2,...,q_n\}$.
	The network with vertices $V$ and adjacency matrix
	\begin{equation}
	\Gamma=\left(\begin{array}{ccccccc}
	& -\Gamma_1 & & & & & \\
	-\Gamma_1 &  & & & & \Gamma_X & \\
	& & & \Gamma_1 & & & \\
	& & \Gamma_1 & & &  & -\Gamma_X\\ 
	& &  & & -\Gamma_S & H_Z &\\ 
	& \Gamma_X &  &  & H_Z^T &  & -H_Z^T\\ 
	&  &  & -\Gamma_X &  & -H_Z & \Gamma_S
	\end{array}\right)
	\end{equation}
	with 
	\begin{align}
	(\Gamma_X)_{ij}=x_{ij},\\
	(\Gamma_1)_{ij}=&\delta_{ij} \mathrm{wt}(\bar{X}_i),\\
	\text{and }(\Gamma_S)_{ij}=&\sum_l (H_X)_{il} (H_Z)_{jl}
	\end{align}
	implements an $[[n,k,d]]$ QNC w.r.t. errors on $Q$. Here $\mathrm{wt}(\bar{X}_i)$ denotes the number of qudits on which the $i$-th logical $X$ operator acts non-trivially. The network produces $k$ Bell pairs and thus allows to teleport $k$ qudits from $A$ to $B$. 
\end{theorem}
\begin{proof}
The pre-measurement graph state $\ket{G}$ is stabilized by
\begin{align}
 S_{X_i}=&X_{a_i} X_{b_i} \prod_{j} X_{q_j}^{x_{ij}} \label{eq:SXi},\\
 S_{Z_i}=&Z_{a_i} Z_{b_i} X_{{a_i}'}X_{{b_i}'} \label{eq:SZi}\\
\text{and } S_i=&X_{s_i} X_{s_i'}\prod_j X_{q_j}^{(H_X)_{ij}}. \label{eq:Si}
\end{align}
This can be seen by multiplying the stabilizer generators centred on the qudits of the $X$ operators which are explicitly included in the above equations. Note that two operators
\begin{align}
A=&\prod_i  X_i^{a_{x,i}} Z_i^{a_{z,i}}\\
\text{and }B=&\prod_i  X_i^{b_{x,i}} Z_i^{b_{z,i}}
\end{align}
commute if and only if
\begin{equation}
\sum_i (a_{z,i} b_{x,i}- a_{x,i}b_{z,i}) \bmod D=0.
\end{equation}
The exponent of $Z_{s_j}$ in $S_{X_i}$ is
\begin{equation}
\sum_{k=1}^n x_{ik} (H_Z)_{jk}=0,
\end{equation}
because $\bar{X_i}$ and 
\begin{equation}
\tilde{s}_j=\prod_k X_k^{(H_X)_{jk}}  Z_k^{(H_Z)_{jk}}
\end{equation}
commute.
The matrix $\Gamma_S$ is chosen s.t. the exponent of $Z_{s_j}$ in $S_i$ is
\begin{equation}
-(\Gamma_S)_{ij} + \sum_k (H_X)_{ik} (H_Z)_{jk} = 0.
\end{equation}
While the main stabilizers given in Eqs.~(\ref{eq:SXi}) and (\ref{eq:SZi}) define the post-measurement state (i.e. one reads from these equations, that Bell pairs are produced), the operators defined in Eq.~(\ref{eq:Si}) are the $X$-chains used for error correction. By construction these operators correspond to the stabilizer generators of the error correction code we started with, see explanations following Definition~\ref{def:networkerrorcorrection}.
\end{proof}
\section{The tandem arrangement in the state formalism}\label{app:tandem}
Consider $G=(\{1,2,3\},\{\{1,2\},\{2,3\}\})$ (\raisebox{0.15\baselineskip}{\includegraphics{tandemidentitya}}) with edge weights $1$. $\ket{G}$ is stabilized by $X\otimes \1 \otimes X^{-1}$ and $Z \otimes X \otimes Z$. Now the node $2$ is measured in the $X$-basis. Let the outcome be $H\ket{x}$. We consider the state of the remaining system. Because $g_3=\1\otimes Z\otimes X$, $Z_2$ acts on this state exactly the same as $X_3$. After application of the by-product operator $X_3^{-x}$ we expect that it is stabilized by $X\otimes X^{-1}$ and $Z\otimes Z$. Up to a change of basis we produced a Graph state. A simple calculation in the state formalism confirms these considerations. Let $\omega=\eulere^{2\pi \ramuno/D}$.
\begin{align*}
 \ket{G}=& \left( \sum_j \ket{j}_1\bra{j}_1 \otimes Z_2^j\right) \left(\sum_k \ket{k}_2\bra{k}_2 \otimes Z_3^k\right)\\& \times \frac{1}{D\sqrt{D}} \sum_{i_1i_2i_3}\ket{i_1i_2i_3}\\
=& \frac{1}{D\sqrt{D}} \sum_{i_1i_2i_3} Z_2^{i_1} Z_3^{i_2}\ket{i_1i_2i_3}\\
=& \frac{1}{D\sqrt{D}} \sum_{i_1i_2i_3}\omega^{i_1i_2} \omega^{i_2i_3}\ket{i_1i_2i_3}\\
=& \frac{1}{D\sqrt{D}} \sum_{i_1i_2i_3}\omega^{i_2(i_1+i_3)} \ket{i_1i_2i_3}\\
=& \frac{1}{D\sqrt{D}} \sum_{i_1i_3} X_3^{i_1+i_3} \ket{i_1}_1\ket{-i_1}_3 Z_2^{i_1+i_3} \sum_{i_2}\ket{i_2}_2\\
\intertext{Now measure and obtain the $(x:=i_1+i_3)$-th eigenvalue and apply the by-product operator $X_3^{-x}$}
\ket{G'}=& \frac{1}{\sqrt{D}} \sum_{i_1} X_3^{-x} X_3^{x} \ket{i_1\,-i_1}\otimes H\ket{x}\\
        =& \frac{1}{\sqrt{D}} \sum_{i_1} \ket{i_1\,-i_1}\otimes H\ket{x}
\end{align*}
Trace out qudit 2 for clarity. The remaining state is stabilized by $Z_1 Z_3$ and $X_1 X_3^{-1}$, so it is, up to the basis transformation $H_3^\dagger$, 
the graph state associated with $G'=(\{1,3\},\{\{1,3\}\})$ (i.e. a Bell pair).
\end{document}